\documentclass[a4paper,twocolumn,british,unpublished]{quantumarticle}
\pdfoutput=1

\usepackage{header}

\newcommand*\otau{\overline{\tau}}
\newcommand*\msc{\mathrm{sc}}
\newcommand*\mtc{\mathrm{tc}}

\begin{document}


\makeatletter
\let\stashsmartcomma\sm@rtcomma
\let\sm@rtcomma,
\makeatletter

\title{Quantum Circuit Optimisation and MBQC Scheduling with a Pauli Tracking Library}
\author{Jannis Ruh}
\email[]{jannis.ruh1@uts.edu.au}
\orcid{0009-0004-9820-7348}
\affiliation{Centre for Quantum Software and Information, School of Computer Science,
Faculty of Engineering \& Information Technology, University of Technology Sydney, NSW
2007, Australia}
\author{Simon Devitt}
\email{simon.devitt@uts.edu.au}
\orcid{0000-0002-5901-1391}
\affiliation{Centre for Quantum Software and Information, School of Computer Science,
Faculty of Engineering \& Information Technology, University of Technology Sydney, NSW
2007, Australia}
\affiliation{InstituteQ, Aalto University, 02150 Espoo, Finland.}


\begin{abstract}
  We present a software library for the commutation of Pauli operators through quantum
  Clifford circuits, which is called Pauli tracking. Tracking Pauli operators allows one
  to reduce the number of Pauli gates that must be executed on quantum hardware. This is
  relevant for measurement-based quantum computing and  for error-corrected circuits that
  are implemented through Clifford circuits. Furthermore, we investigate the problem of
  qubit scheduling in measurement-based quantum computing and how Pauli tracking can be
  used to capture the constraints on the order of measurements.
\end{abstract}

\maketitle

\makeatletter
\let\sm@rtcomma\stashsmartcomma
\makeatother

\section{Introduction}\label{s"introduction}

Realising fault-tolerant quantum gates and qubits is costly in terms of space, time and
energy. It is crucial to reduce the number of gates as much as possible. This optimisation
can be tackled on multiple levels, e.g., optimising the quantum algorithm on the
logical circuit level, similar to how classical compilers optimise classical
algorithms, or trying to develop more efficient error correcting codes.

Another problem that appears in the context of \ac{mbqc}
\cite{raussendorf_one_way_quantum_computer,danos_measurement_calculus}, but partially also
in \ac{qec}, are dynamic corrections that are induced into the computation because of the
nondeterminism of the measurements.
For an efficient computation, these corrections have to be actively accounted for, and
cannot be ignored through post-selection.
They are also an essential piece of information for qubit scheduling,
i.e., re-using qubits after they have been measured, as it defines an order on the
measurements, which we shall explain in more detail.

In this paper, we shall focus on the {\em classical tracking} of Pauli operators through a
quantum circuit, called {\em Pauli tracking}
\cite{knill_quantum_computing_with_realistic_noisy_devices,paler_software_based_pauli_tracking,riesebos_pauli_frames,jin_ho_pauli_frames},
which deals with the second problem introduced above and (partially) with the first
problem. Pauli tracking, we shall explain it below, is an optimisation done on the
software level that directly reduces the number of Pauli gates that must be executed on
the quantum hardware. Furthermore, it can capture the information of dynamic Pauli
correction induced by measurements. Pauli tracking is applicable whenever the quantum
circuit comprised of Clifford gates (and arbitrary measurements), as for example, in
\ac{mbqc}, but also in the context of \ac{qec}, e.g., the surface code
\cite{bravyi_lattice_code,fowler_surface_code}. The \ac{qec} codes themselves are often
Clifford circuits on the level of the (uncorrected) physical qubit; but also on the
logical fault-tolerant level, the non-Clifford gates are usually realised via injection of
certain ``magic'' states or specific measurements with the help of additional ancilla
qubits that are entangled using only Clifford gates (cf. \cref{f"teleportation})
\cite{gottesmann_quantum_teleportation_is_universal,zhou_methodology_for_quantum_logic_gate_construction,fowler_surface_code}.
Therefore, Pauli tracking can be applied on both levels in \ac{qec}, i.e., on the physical
level and the logical level, which gives it together with \ac{mbqc} a wide range of
applications.

Let us now sketch out what we mean with Pauli tracking and how it works; it is based on
the mathematical foundations discussed in
\cite{bolt_clifford_stuff_1,bolt_clifford_stuff_2,bolt_clifford_stuff_3}. The Clifford
group is the normaliser of the Pauli group, meaning that conjugating the Pauli group with
Clifford operators preserves the Pauli group. This implies that Pauli operators can be
classically efficiently, i.e., without exponential costs, commuted with Clifford
operators. A central implication of that characteristic is the Gottesman-Knill theorem
\cite{gottesman_heisenberg_representation_of_quantum_computers}, which states that
circuits consisting only of Clifford gates can be efficiently simulated with a classical
computer via a stabiliser simulator, e.g., \cite{gidney_stim}. For universal quantum
computation, however, the Clifford group is not sufficient, i.e., we need additional
non-Clifford gates. In this case, stabiliser simulators are not efficient anymore,
however, for certain realisations of quantum computers, we can still make use of the fact
that the Pauli group is preserved under conjugation of Clifford operators. For example, in
\ac{mbqc} or in the context of \ac{qec}, where the circuits only consist of Clifford gates
and measurements, it is possible to traverse the Pauli operators through the circuit by
commuting them with the Clifford gates until a measurement is reached and then account for
the Pauli operators through post-processing or adaption of the of the measurement (cf.
\cref{f"teleportation} (c)). This effectively reduces the number of Pauli gates that have
to be executed on the quantum hardware to the {\em number of qubits or less}.

In \cref{s"pauli_tracking_library}, we present a {\em library} that can be used to perform
the Pauli tracking \cite{pauli_tracker_software}. It is a low-level library, natively
written in Rust, with a Python wrapper and a partial C interface. The library is designed
to be dynamically used when compiling, or executing quantum circuits and supports various
generic data structures for different use cases. See \cref{s"software_availability} for
how to access the library.

Now when teleporting non-Clifford gates in the context of \ac{qec}, or in general any gate
as in \ac{mbqc}, the non-determinism of the according measurements usually introduces
Pauli corrections (or in general Clifford corrections) conditioned on the measurement
outcomes, as for example, depicted in \cref{f"teleportation}. These corrections
effectively define a strict partial time order for the execution of the circuit, because
of their non-determinism in general. The tracking of the Pauli corrections through the
circuit directly captures this order (except for the order induced by possible Clifford
corrections, however, they can also be deferred with an additional teleportation; cf.
\cref{f"teleportation} (c)) and reduces it to the measurements. Analysing this
information, in connection with entanglement structure of the circuit or graph state,
gives knowledge about how the computation can be optimised in space and in time, i.e.,
what is the minimal number of required qubits (fully space optimised), what is the minimal
number of steps of parallel measurements (fully time optimised), or in general, given a
fixed number of qubits, how many steps of parallel measurements are required. Furthermore,
capturing this time order and the corrections is necessary for certain \ac{mbqc}
compilation strategies and frameworks, e.g., Refs.
\cite{vijayan_compilation_of_algorithm_specific_graph_states,le_benchq}, where a logical
circuit is transformed into a logical graph state (using teleportation techniques and
stabilizer simulations) and the captured time order defines how the graph state is
consumed.

In \cref{s"mbqc_scheduling}, we shall focus on this problem. More specifically, we define
the {\em scheduling problem} for \ac{mbqc}, which may also be applied in the context of
\ac{qec} circuits locally at the parts where non-Clifford gates are teleported, and tackle
it detached from the underlying quantum circuit on the logical level (i.e., we do not
consider the quantum hardware). The problem will be defined in a way that is applicable to
\ac{mbqc} protocols like the Raussendorf cluster
\cite{raussendorf_one_way_quantum_computer}, but also other protocols that involve more
complex graph states, e.g., Refs.
\cite{vijayan_compilation_of_algorithm_specific_graph_states,le_benchq}, or in general,
whenever the entanglement structure of the circuit can be described with a graph and
teleportation techniques are used for non-Clifford gates. However, note that the algorithm
for the space-time optimisation that we provide is only scalable to a certain degree due
to the hardness of the problem (the problem is related to finding the pathwidth of a graph
\cite{elman_scheduling_via_path_decompositions}). This can be resolved if the input can be
split into smaller problems (e.g., one might be able to serially split the circuit
\cite{le_benchq}), but if the goal is to optimise larger inputs, more optimised algorithms
may be required. However, as a preview, note that finding a time-optimal solution is
independent of the underlying graph, and therefore not restricted to pure \ac{mbqc}, and
can be directly calculated from the Pauli tracking in polynomial time.
\vspace*{-0.2cm}
\begin{figure}[H]
  \center
  \includegraphics[scale=0.790]{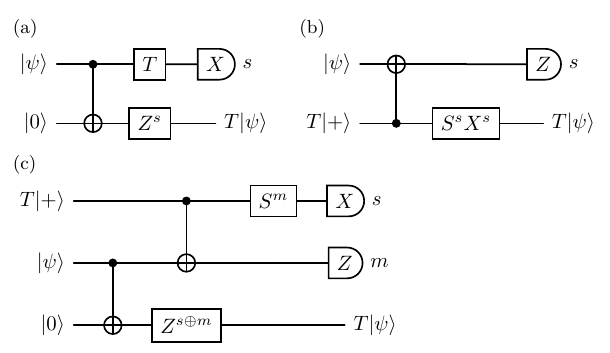}
  \caption[Example of $T$ gate teleportation]{
    Example of $T$ gate teleportation; these are typical protocols in
    \ac{mbqc} but also especially in \ac{qec} where the ancilla states, $T\ket{+}$, are
    prepared in a distillation process \cite{campbell_distillation_framework}. The input
    state $\ket{\psi}$ is teleported to the output state $T\ket{\psi}$. (a) The $T$
    teleportation is achieved with a ``magic'' measurement $TXT\ur$. This introduces a
    Pauli correction $Z$ on the output qubit depending on the measurement result. The
    correction can be tracked as a Pauli frame through a subsequent Clifford+measurement
    circuit, unblocking the execution until a measurement is reached for which the Pauli
    frame defines a non-trivial correction. (b) Implementation of the $T$ teleportation
    through injection of a ``magic'' state $T\ket{+}$. This protocol introduces a
    potential $S$ (phase gate) correction on the output qubit. Since this correction
    cannot be tracked through the subsequent circuit, it is blocking, i.e., the execution
    on the output qubit has to wait until the measurement result $s$ is known. (c)
    Implementing the $T$ teleportation with a magic state injection, however, without
    blocking the execution. This circuit is constructed by taking circuit (a) and then
    implement the $T$ gate there with circuit (b). This way the blocking $S$ correction is
    removed from the output qubit. The $Z^m$ correction on the output qubit can be
    obtained through Pauli tracking: \textit{commute} the $X$ correction in circuit (b)
    through the $S$ correction, turning it into a $ZX \propto Y$ correction; the $X$ part
    does not matter for the following $X$ measurement and can be completely
    \textit{removed}; the $Z$ part, however, flips the measurement result $s$, which can
    be accounted for by \textit{moving} the $Z$ correction onto the output qubit as a $Z$
    correction (since $s$ induces a $Z$ correction). All of these three operations,
    \textit{commute}, \textit{remove}, and \textit{move}, are supported in our Pauli
    tracking library.
  }
  \label{f"teleportation}
\end{figure}
\vspace*{-1.2cm}

\section{The Pauli Tracking Library}\label{s"pauli_tracking_library}

In \Cref{s"introduction}, we already sketched out the basic idea. Here, we shall focus on
a mathematical description which directly represents how the library is implemented.
Afterwards, we discuss some features of the library.

\subsection{Mathematical Formulation}\label{s"mathematical_formulation}

The Pauli operators, $X, Y, Z \in \U\w(\C^2)$, where $\U$ denotes the unitary group, are
defined as
\begin{equation*}
  X \coloneq \matp{0&1\\1&0} \quad
  Y \coloneq \matp{0&-i\\i&0} \quad
  Z \coloneq \matp{1&0\\0&-1}\;,
\end{equation*}
and we define the Pauli group as follows:
\begin{definition}[Pauli group]
Let \en. The Pauli group $\gp_n \leq \U\w(\C^{2^n})$ is defined by its generators via
\begin{equation*}
  \gp_n \coloneq \braket{i, X_1, Z_1, \ldots X_n, Z_n} \;.
\end{equation*}
For later reference, we also define the (Heisenberg-Weyl)
group $\heip_n \coloneq \braket{X_1, Z_1, \ldots X_n, Z_n}$.\\
The projective groups are $\op_n = \gp_n/\braket{i}$ and $\oheip[n] =
\heip_n/\braket{-1}$, respectively.
\end{definition}
One can also define the Pauli group differently, e.g., by including phases or arbitrary
complex pre-factors, but note that the respective projective groups, which are the groups
we are interested in, are all isomorphic, e.g., it is $\op_n \cong \oheip[n]$.

The Clifford group is now defined as the normaliser of the Pauli group:
\begin{definition}[Clifford group]\label{.clifford_group}
Let \en. The (unitary) Clifford group is defined by
\begin{align*}
  \gc_n \coloneq \set{U \in \U\w(\C^{2^n})}{U\gp_n U\ui = \gp_n} \;.
\end{align*}
The projective group is given by $\oc_n = \gc_n/\U(1)$.\\
Alternatively, the Clifford group can be defined through its generators, e.g.
\cite{selinger_normal_form_of_cliffords},
\begin{equation*}
  \gc_n = \groupset{u, H_i, S_i, \cz[ij]}{u \in \U(1);\; i, j \in \N_{\leq n}; \;
  i \neq j}\;,
\end{equation*}
where $H$ is the Hadamard operator $H
= \frac{1}{\sqrt{2}}\smatp{1&1\\1&-1}$, $S$ the phase operator $S = \diag\w(1, i)$ and
$\cz$ the controlled-$Z$ operator $\cz = \diag\w(1, 1, 1, -1)$.
\end{definition}
Again, we could have chosen slightly different definitions, but the resulting groups would
be either the same or at least the projective groups would be isomorphic\footnote{For
example, one could replace the unitary group $\U$ with the linear group $\GL$, which
would result in a Clifford group $\ic_n \leq \C^\times \U\w(\C^{2^n})$, i.e., more
specifically, the same group as before, but with arbitrary pre-factors $\C^\times = \C
{\setminus} \bc{0}$, cf. \cite{bolt_clifford_stuff_1,bolt_clifford_stuff_2}. Using
$\C^\times \gp_n$ as normalised group, leads to the same Clifford group, however, using
$\heip_n$ as normalised group would lead to different group, since it is not invariant
under Clifford conjugations.}.

Now let us assume we have two Pauli operators $p_1, p_2 \in \gp_n$ and a Clifford operator
$c \in \gc_n$ that form the circuit $p_1 c p_2$. Under Pauli tracking, we understand the
process of transforming this circuit into $p c$ for some $p \in \gp_n$, that is we
traverse all Pauli operators to the end of the circuit and collapse them into one Pauli
operator. It is clear that this is achieved by setting $p = p_1 c p_2 c\ui$, i.e.,
conjugating the second Pauli operator with the Clifford operator. However, we can reduce
the problem to a simpler form. Firstly, since we are dealing with quantum circuits, scalar
factors do not matter, and we can reduce the Pauli group to its projective group, i.e.,
set $p_1, p_2, p \in \op_n$. Secondly, we also only have to consider the projective
Clifford group, i.e., set $c \in \oc_n$, since scalar factors are cancelled when
conjugating. Moreover, since conjugating Pauli operators with Pauli operators only
introduces a phase (more specifically, the centraliser of $\op_n$ in $\oc_n$ is $\op_n$
\cite{bolt_clifford_stuff_2}), we can reduce the Clifford group to the quotient group
$\oc_n / \op_n$. This leads us to the following definition of the Pauli tracking problem:
\begin{definition}[Pauli tracking]\label{.pauli_tracking}
Let \en[m, n]. Given a sequence $\w(g_i)_{1 \leq i \leq m} \subseteq \op_n \cup \oc_n /
\op_n$, calculate $\w(p_i)_{1 \leq i \leq m} \subseteq \op_n$, which is recursively
defined by
\begin{equation}
  p_i = \bcdot{\begin{array}{ll}
    g_i p_{i-1} & \text{, if } g_i \in \op_n\\
    g_i p_{i-1} g_i^{-1} & \text{, if } g_i \in \oc_n / \op_n
  \end{array}}
\end{equation}
for $i \in \bc{1, \ldots, m}$, where $p_0 = \1$.\\
We refer to $\w(p_i)_{0 \leq i \leq m} \subseteq \op_n$ as the Pauli frame sequence that
is tracked through the circuit sequence $\w(g_i)_{1 \leq i \leq m}$.
\end{definition}
This is the computational task we want solve with our library, which is essentially
achieved by implementing the following two isomorphisms.
\begin{proposition}\label{.quotient_tableau_description}
Let \en. $\op_n$ is isomorphic to the abelian group $\overline{\heib[n]} =
\w(\Z_2^n \times \Z_2^n, +) \cong \Z_2^{2n}$ with the standard addition via
\begin{align}
  \otau: \w(\Z_2^n \times \Z_2^n) &\to \op_n, \quad (z, x) \mapsto
  \bigotimes_{j=1}^n \overline{Z}_j^{z_j} \overline{X}_j^{x_j}\;.
\end{align}
\end{proposition}
\begin{proposition}[\cite{bolt_clifford_stuff_1, bolt_clifford_stuff_2}]
\label{.clifford_symplectic_isomorphism}
Let \en. The projective Clifford group, up to Pauli operators, is isomorphic to the
symplectic group, i.e.,
\begin{equation*}
  \kappa : \oc_n/\op_n \to \Sp_{2n}\w(\Z_2), \quad c\op_n \mapsto
  S_{c} = \otau\ui \circ \inn_c \circ \otau\;,
\end{equation*}
where $\Sp_{2n}\w(\Z_2)$ is the symplectic group of the $\Z_2^{2n}$ vector space with
respect to the standard symplectic form $\smatp{0&\1\\\1&0}$, and $\inn_c$ is the inner
automorphism induced by $c$, i.e., conjugation with $c \in \oc_n$.
\end{proposition}
The first isomorphism $\otau$ in \cref{.quotient_tableau_description} describes how Pauli
operators are represented in our library, i.e., they are simply represented by a pair of
boolean values or bits and multiplication is done via the XOR operation. The second
isomorphism $\kappa$ in \cref{.clifford_symplectic_isomorphism} then says that the
according operation of a conjugation with a Clifford element $c\op_n$ in the binary Pauli
representation is given by the symplectic operator $\kappa\w(c\op_n)$. Reference
\cite{conjugation_rules} explicitly lists the symplectic operators for some Clifford
elements that are implemented in the library.

Note that the process of Pauli tracking is very similar to stabiliser simulations. The
group $\heib[n]$ is the Heisenberg group and $\heip_n$ is its isomorphic Weyl
representation
\cite{bolt_clifford_stuff_1,bolt_clifford_stuff_2,bolt_clifford_stuff_3,gross_hudsons_theorem_finite_quantum_system,kibler_heisenberg_pauli_weyl,tolar_cliffords_in_quantum_computing}.
In stabiliser simulations, these isomorphisms are adjusted to non-trivial stabiliser
subgroups of  $\gp_n$ and the updating of the stabilisers works similar to the Pauli
tracking, with the differences that signs have to be accounted for.

The matrix representation of the Clifford conjugations gives us an upper bound for the
computation cost of the Pauli tracking per gate, that is, the cost of one conjugation is
bounded by $\order[n^2]$. However, in reality, the cost per gate is usually much lower,
since the standard Clifford gates used in quantum circuits are usually local to a lower
number of qubits \en[m], $m \leq n$, and they can often be implemented in a more efficient
way than simple matrix multiplication. For example, a Hadamard gate is often just one
memory swap, or a controlled-$Z$ gate can be implement with two XOR operations. If we
bound $m$ by a constant, for example, $m = 2$, and the circuit consists of \en[l] gates,
and we track \en[k] Pauli frames simultaneously, the Pauli tracking can be performed in
$\order[lk]$ time with $\order[nk]$ memory. In the case of \ac{mbqc}, under the
assumption that each teleportation induces a Pauli correction, that is, a Pauli frame (cf.
\cref{f"teleportation}), it is roughly $l \propto k \propto n$, i.e, both costs, time and
space, are of the order $\order[n^2]$.

\subsection{Library Features}

In \cref{s"mathematical_formulation}, we focused on how the Pauli tracking logic is
implemented. In this section we shall give a brief overview of how the library can be
used. However, we shall not give examples or explicitly discuss the API; for that, please
refer to the documentation of the software packages (cf. \cref{s"software_availability}).

There are essentially two modes in which the library can be used: The first one is for
Pauli tracking when all gates are known, for example dynamically during execution of a
quantum circuit. In this mode, there is one Pauli frame (cf. \cref{.pauli_tracking}), that
is, one Pauli operator for each qubit, and this frame is updated accordingly to the
circuit instructions. The required memory for this is linear in the number of qubits.

The second mode is to perform the Pauli tracking when defining the quantum circuits, or
compiling them, that contain gate teleportations, e.g, as in \ac{mbqc}. In this mode, for
each potentially induced Pauli correction (it is usually non-deterministic since it
depends on a measurement result), one Pauli frame can be captured, and then tracked
through the circuit. These frames can then be used to determine the strict partial time
order of the measurements that we discussed in the introduction. During executing, the
Pauli corrections before the measurements can then be obtained from the frames. In this
mode the required memory is linear in the number of qubits and frames, i.e., induced
corrections. For pure \ac{mbqc} circuits, this means that the memory scales approximately
quadratically with the number of vertices in the graph state. The frames are stored in
major-qubit-minor-frame order. This way, the conjugations can be performed through
vectorised operations and simple memory swaps.

The way the Pauli tracking is performed from the user side, is to initialise a tracking
object, and then update it by calling the according methods that correspond to the circuit
instructions (cf. library documentation examples). This works, for example, similar to how
quantum circuits are constructed in some of the quantum computing libraries, e.g., Circ
\cite{circ} or Qiskit \cite{qiskit}, where one initialises a circuit object and then adds
gates to it.

When using the Rust native library, the user can choose among different data structures
for the representation of the Pauli frames. This is achieved by designing the library
generically (through static dispatch, i.e., monomorphization) in its core data structures.
This way, the user can choose the most appropriate data structure for the specific use
case. For example, while \acs{simd} bit-vectors (\acl{simd}) may allow faster executions
of the Clifford conjugations, normal bit-vectors can be more efficient if the user often
has to access the Pauli frames. The library directly supports some standard data
structures, but the user can also support their own data structure by providing the
required methods. For example, to support a data structure as top-level tracking object,
the data structure only has to provide methods that implement the $H$, $S$, and $\cz$
gates since every other Clifford gate can be constructed through default methods according
to \cref{.clifford_group}.

When using the Python wrapper, this generic flexibility is not completely given, but
provided to a restricted set of data structures.

Finally, we would like to briefly discuss where this library stands in relation to other
software tools that do provide similar functionality, for example, Stim
\cite{gidney_stim}. Stim is a stabiliser simulator which internally tracks Pauli operators
to simulate error models. A subproject of Stim also provides a graphical application to
visualise how Pauli errors propagate through a Clifford circuit. Our library does not try
to directly compete with these existing projects, w.r.t. their Pauli tracking
functionality. Instead we want to provide flexible, low level functionality, which focuses
solely on the Pauli tracking, that can be integrated easily into multiple other projects.
The goal is to provide an API that allows for simple integration with minimal overhead.

\section{MBQC scheduling}\label{s"mbqc_scheduling}

In \ac{mbqc} \cite{raussendorf_one_way_quantum_computer,danos_measurement_calculus}, the
quantum gates are realised through gate teleportation protocols, i.e., entanglement of the
qubits with additional qubits and then performing specific measurements which effectively
realise the gate. The entangled resource state is usually describe by a graph where the
vertices represent the qubits and the edges the entanglement. The graph usually contains a
large number of vertices, i.e., qubits, and one critical aspect of \ac{mbqc} is to reuse
qubits after they have been measured. More specifically, it is important to schedule
initialisation, entanglement and measurement of the qubits in way that reduces the quantum
memory requirement (space cost) and the execution time (time cost). This, of course,
depends heavily on the underlying hardware and its architecture, but even without
considering that, it is a hard problem.

The scheduling is also restricted by certain constraints. One of the constraints is
induced by the non-determinism of the measurements. Since the measurement outcome is not
known prior to execution, they introduce Pauli corrections depending on the measurement
outcome. If these corrections commute or anticommute with the subsequent measurement, they
can be accounted for by post-processing, however, this is not the case in general.
Therefore, the corrections define a time order for the measurements. For certain \ac{mbqc}
protocols, like the Raussendorf two-dimensional cluster
\cite{raussendorf_one_way_quantum_computer}, a correct scheduling of the measurements is
directly included in the protocol and graph state, however, this scheduling may not be
time optimal. Moreover, for other protocols, that for example involve certain graph
transformations, finding a time order that is not too restrictive may be decoupled from
the graph state. This is where the Pauli tracking can help. For example, consider the
protocol proposed in
\cite{vijayan_compilation_of_algorithm_specific_graph_states,le_benchq}: In this protocol,
a quantum circuit is first transformed into a larger Clifford circuit that includes gate
teleportations, and then a graph state is computed from the Clifford circuit. The graph
state is independent of the Pauli corrections, however, they are captured via Pauli
tracking during the transformation to the Clifford circuit. From the captured Pauli
frames, one can then calculate the measurement time order for the graph, simply by
checking for each frame on which qubits it induces potential corrections.

The other constraint on the scheduling is that a qubit can only be measured when all its
neighbours have been initialised and entangled with the qubit. This is because the
entanglement does not commute with the measurement.

In the following, we formulate a framework that describes the scheduling problem. We shall
then describe an attempt to solve this problem up to a certain scale and present some
numerical results that give indications about how much can be gained with this
optimisation \cite{mbqc_scheduling_software}.

\subsection{Framework}

We start by defining a valid measurement schedule that accounts for the constraints
described above. The graphs are undirected and simple, i.e., no self-loops or multiple
edges between the same two vertices are allowed.
\begin{definition}[Measurement schedule]\label{.measurement_schedule}
Let $\w(V, E, \prec)$ be a triple where $G = \w(V, E)$ is a graph (with vertices $V$ and
edges $E$) and $\prec$ is a strict
partial order on $V$. A measurement schedule is a sequence $S = \w(M_i, I_i)_{1\leq i \leq
n}$, where $M_i, I_i \subseteq V$ and \en, such that for all $1 \leq i \leq n$ (negative
indexed sets are empty)
\begin{enumerate}
\item\label[condition]{i"schedule_init} $M_i \cup \neigh\w(M_i) \subseteq I_i$,
\item\label[condition]{i"schedule_order} $\mathrm{PRE}\w(M_i) \subseteq \bigcup_{1\leq j < i} M_j$,
\item\label[condition]{i"schedule_all} $\dot\bigcup_{1\leq i \leq n} M_i = V$,
\item\label[condition]{i"schedule_keep} $I_{i-1} {\setminus} M_{i-1} \subseteq I_i \subseteq
  V {\setminus} \bigcup_{1\leq j < i} M_j$,
\end{enumerate}
where $\neigh\w(M_i) = \set{\neigh(x)}{x \in M_i}$ is the set of all neighbours of the vertices in $M_i$
and $\mathrm{PRE}\w(M_i) = \set{x \in V}{\exists y \in M_i : x \prec y}$ is the set of all
vertices that are smaller than a vertex in $M_i$.
\end{definition}
The first two \cref{i"schedule_init,i"schedule_order} describe the constraints under which
a vertex can be measured, that is, the vertex itself and all its neighbours have to be
initialised (and entangled), and all smaller vertices that induce possible corrections on the
to be measured vertices have to be measured before. The last two
\cref{i"schedule_all,i"schedule_keep} are just for soundness, i.e., all vertices have to be
measured and each vertex is only initialised once and kept initialised (and entangled)
until it is measured.

We can now define the space and time cost of a schedule, i.e., how much quantum memory is
required and how many steps of parallel measurements are executed:
\begin{definition}[Schedule cost]\label{.schedule_cost}
Let $\w(V, E, \prec)$ be a graph with a strict partial order and $S = \w(M_i, I_i)_{1\leq
i \leq n}$ a corresponding measurement schedule, \en. The space cost, $\msc$, and time
cost, $\mtc$, of the schedule are defined by
\begin{align*}
  \msc\w(S) &= \max_{1\leq i \leq n} \abs{I_i}\;,\\
  \mtc\w(S) &= n\;.
\end{align*}
\end{definition}
The optimisation problem, one is now faced with, is to find a schedule $S$ with small
space and time cost. For the special case that there is effectively no time order, i.e.,
there are no $\prec$ relations, it has been shown that finding a space-optimal schedule is
equivalent to finding a path decomposition of the graph that realises the pathwidth
\cite{elman_scheduling_via_path_decompositions}, implying that the pathwidth is a lower
bound for the space cost $\msc\w(S)$. Just approximating the pathwidth is already NP-hard
\cite{elman_scheduling_via_path_decompositions}.

Given a graph and a time order, there are many possible schedules, e.g., given a sequence
$\w(M_i)_i$ there are many possible $\w(I_i)_i$, such that $S = \w(M_i, I_i)_i$ is a
valid measurement schedule. However, we are only interested in \textit{minimal}
$\w(I_i)_i$ w.r.t. a \textit{measurement pattern} $\w(M_i)_i$:
\begin{definition}[Measurement pattern]\label{.measurement_pattern}
Let $\w(V, \prec)$ be a tuple of a set of vertices $V$ with a strict partial order $\prec$. A
measurement pattern is a sequence $P = \w(M_i)_{1\leq i \leq n}$, where $M_i \subseteq V$
and \en, such that for all $1 \leq i \leq n$
\begin{enumerate}
\item $\mathrm{PRE}\w(M_i) \subseteq \bigcup_{1\leq j < i} M_j$,
\item $\dot\bigcup_{1\leq i \leq n} M_i = V$.
\end{enumerate}
\end{definition}
In contrast to a \textit{measurement schedule}, a \textit{measurement pattern} depends
only on the time order but not on the spatial structure of the graph. However, we can
construct a unique \textit{minimal} schedule with respect to the pattern:
\begin{proposition}\label{.construction}
Let $\w(V, E, \prec)$ be a graph with a strict partial order and $P = \w(M_i)_{1\leq i
\leq n}$ a measurement pattern of $\w(V, \prec)$, \en. Then the schedule $S = \w(M_i,
I_i)_{1\leq i \leq n}$ with
\begin{equation}
  I_i = M_i \cup \neigh\w(M_i) \cup \w(I_{i-1} {\setminus} M_{i-1})
\end{equation}
for all $1 \leq i \leq n$ (negative indexed sets are empty) is space optimal w.r.t. $P$ in
the sense that the measurement sets $M_i$, $1 \leq i \leq n$, are the same for $P$ and
$S$, and the space cost is minimal.
\end{proposition}
\begin{proof}
First, we need to show that $S$ is a valid schedule. The only non-trivial condition is
$I_i \subseteq V {\setminus} \bigcup_{1\leq j < i} M_j$, for all $1 \leq i \leq n$, but this
follows easily by induction.

The relative space optimality now simply follows from the fact that $I_i$ is obviously the
minimum set that satisfies \cref{i"schedule_init,i"schedule_keep} for all $1 \leq i \leq
n$.
\end{proof}

\subsection{Algorithm}

In the following two sections, we investigate the usefulness of scheduling optimisations
on smaller scales, i.e., how much can the time and space costs be potentially reduced
depending on the graph structure and the time order. In particular, the goal is not to
provide an algorithm that can be used for large scale quantum computing. However, note
that finding the time-optimal schedule can be achieved in polynomial time, as we describe
at the end of this section. Since finding the space-optimal scheduling is a
computationally hard problem, real applications will usually not reach the lower bound,
however, the presented results give insights on how much improvement on the space costs by
scheduling optimisations may be expected.

Our algorithm chooses different measurement patterns $P$, i.e., a sequence of parallel
measurement steps that are allowed by the time order, and then creates a measurement
schedule $S$ according to \cref{.construction}. For this schedule it calculates then the
space and time costs. The hardness of the problem lies in the sheer number of possible
choices for $P$. Again, considering the extreme case that there is no time order, the
possible choices for $P$ are the ordered partitions of the set $V$, and the number of
those partitions is given by the ordered Bell number, which is approximately $\abs{V}! /
\w(2 (\log 2)^{\abs{V}+1})$, asymptotically in $\abs{V}$
\cite{bailey_number_of_weak_orderings}.

However, not all patterns have to be considered necessarily. In our algorithm, the
measurement patterns are dynamically created, sketched out in the following (for more
details, please refer to the source code; it is essentially a depth-first-search): Let us
assume we already calculated the space and time costs for some schedules. Now when
choosing a new pattern, we first choose a first measurement set $M_1$. For this set, we
construct $I_1$ and the calculate the space and time costs so far (i.e., $\abs{I_1}$ and
$1$, respectively). If these costs are worse than the costs for the best schedules we
found so far, we directly discard the pattern and try a new set $M'_1$. If the costs are
better, we continue and choose a second set $M_2$. Again, we calculate the costs so far
and either discard the second set or continue. We do this until we finally have a complete
pattern and then repeat the process for a new pattern (in reality, we are not actually
always starting with a first set $M_1$, but rather go a few steps back and then forward
again with different measurement sets). This technique
allows us to potentially skip many possible patterns, however, it highly depends on the
structure of the graph and the time order. The scaling is in general still expensive
though. We tried to calculate the complexity of this implementation, but failed to derive
a simple expression; \cref{f"runtime} shows the run-times for some calculations, however,
note that this is not a benchmark but only qualitative information.

The algorithm described so far, covers only an exact optimisation, i.e., finding the
minimum costs, however, often approximations are sufficient. We implement an approximative
version of the algorithm by putting a probabilistic condition on the acceptance of a next
measurement set for a pattern, that is, if we choose some set $M_i$ and the costs are
better than the best costs so far, the set is still only accepted with a certain
probability. The probability function depends on appropriate parameters and can be
specified by the user. This allows us to reduce the run-time, however, the final results
are not guaranteed to be optimal, but only approximate the minimum costs, since some
optimal patterns might be discarded (and depending on how aggressively measurement sets
are discarded it might still not be scalable).

To get a time-optimal measurement schedule, however, is trivial, given a time order. For
the next measurement set, one simply chooses all vertices that are allowed to be measured.
Given some Pauli frames, the Pauli tracking library provides a method to calculate the
according time order in $\order[\abs{V}^3]$ time, representing it as a reduced directed
acyclic graph. In this graph, the qubits are sorted into layers $\w(M_i)_i$ that describe
when they can be measured the earliest, i.e., the layers define the measurement pattern $P
= \w(M_i)_i$ for the trivial time-optimal schedule (cf. \cref{.measurement_pattern}).
Obtaining the Pauli frames via Pauli tracking from the \ac{mbqc} circuit is bounded by
$\order[\abs{V}^2]$, i.e., the schedule is obtained in $\order[\abs{V}^3]$ time during
compilation.

\subsection{Numerical Results}

\begin{figure*}[hbt!]
  \center
  \includegraphics[scale=0.97]{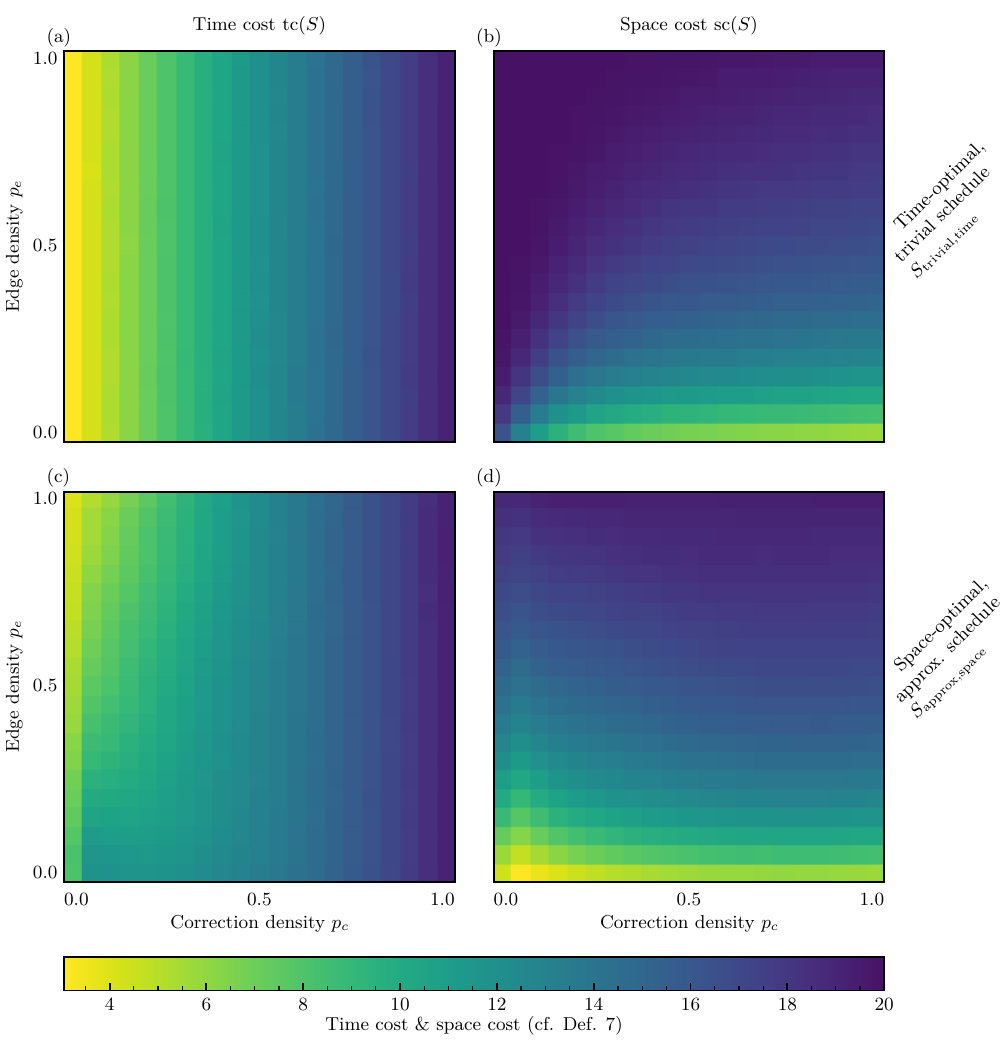}
  \caption[
    Time and space cost for the trivial time-optimal schedule and the approximated
    space-optimal schedule, respectively, for $20$ vertices.
  ]{
    Time and space costs (cf. \cref{.schedule_cost}) for the trivial time-optimal schedule
    and the approximated space-optimal schedule, for $\abs{V} = 20$ vertices. (a) shows
    the time cost $\mtc$ for the time-optimal schedule $S_{\text{trivial,time}}$ directly
    obtained from the time order induced by the Pauli frames. Up to numerical errors, the
    time cost is independent of the edge density $p_e$ and increases with the correction
    density $p_c$. (b) shows the space cost $\msc$ for the same schedule
    $S_{\text{trivial,time}}$. (c) and (d) show the according costs for the approximated
    space-optimal schedules $S_{\text{approx,space}}$, respectively. For details about the
    probabilistic approximated search , look at \cref{s"more_on_the_numerical_results}.
  }\label{f"density}
\end{figure*}

In \cref{f"density,f"nodes_main} we show some numerical results of the space and time
costs for random graphs and time orders. The results are not directly based on quantum
circuits but rather on random graphs and random time orders. The graphs are randomly
generated by uniformly creating edges with a certain density $p_e$. Instead of directly
drawing a random time order, we draw random Pauli frames and then calculate the time order
from them using the Pauli tracking library. We do this, to create a closer connection to
the underlying \ac{mbqc} scheduling problem: First, we randomly pick one vertex, which
represents a vertex/qubit in a teleportation protocol that is going to be ``measured''
(cf. \cref{f"teleportation}). Then, we randomly induce Pauli corrections on the other
vertices (which have not been ``measured'' yet), dependent on the picked vertex, with a
certain probability $p_c$. These corrections form a random Pauli frame that represents
the correction induced by the picked vertex, which would have been tracked through a
hypothetical circuit. This way, one might argue that the correction density $p_c$
resembles the spreading of the induced Pauli corrections because of entanglement.

Alternatively, to have a closer relation to quantum algorithms, we could have directly
drawn random circuits, or maybe specific circuits, and then use different \ac{mbqc}
protocols to transform them into graphs and time orders. However, this compilation is not
trivial to implement and not part of this project.

Keeping this in mind, the shown results should be viewed as qualitative indications
whether it might be worth to invest in this optimisation technique. They are definitely
not benchmarks, since sampling enough data in a stable environment would require extremely
long calculations due to the hardness of the problem. Furthermore, in an application, it
will also heavily depend on the specific quantum algorithm, the \ac{mbqc} protocol, and
the underlying hardware which puts additional constraints on the scheduling.

\Cref{f"density} shows the time cost for the trivial time-optimal schedule and the space
cost for an approximated space-optimal schedule, for a fixed number of vertices, $\abs{V}
= 20$, but different edge densities $p_e$ and correction densities $p_c$. The time cost
for the time-optimal path (\cref{f"density} (a)) is completely independent of the edge
density $p_e$; this is clear, since it is directly calculated from the time order, without
any reference to the graph. The lower the correction density $p_c$ is, the lower is the
time cost. This is because the time order has less order relations, implying less layers
in the directed time order graph, on average. Regarding the space cost for the
approximated space-optimal schedule, (\cref{f"density} (d)) we see that it increases
rather fast with the edge density $p_e$. The reason for that is that more vertices have to
be initialised when we want to measure a vertex, according to
\cref{.measurement_schedule}~\ref{i"schedule_init}, since the vertex has more neighbours.
With increasing correction density $p_c$, there are less choices to measure first a vertex
with less neighbours (\cref{.measurement_schedule}~\ref{i"schedule_order}), and therefore
the cost increases. The other two plots, \cref{f"density}~(b) and (c), show mainly
artefacts of how the algorithm is implemented: For low correction densities the space cost
of the trivial time-optimal schedule is high because this schedule greedily measures as
many vertices as possible. The time cost for the approximated space-optimal schedule
decreases with the correction density, because the algorithm always first tries the more
time-optimal schedules.

While these plots are not based on real quantum circuits, they confirm the expected
behaviour of the costs and maybe provide some qualitative intuition for when it is worth to
search for a more space-optimal schedule, if this is wanted. For example, if the edge
density is relatively high, e.g., $p_e > 0.8$, and the correction density is not too
low, e.g., $p_c > 0.6$, it might not be worth searching for a space-optimal schedule
because the space cost probably cannot be reduced significantly. On the other hand, if the
edge density is low, the space cost can probably be significantly reduced.

In \cref{f"nodes_main} we can see analogous costs, but now for a varying number of vertices
$\abs{V}$. The edge and correction densities scale via $p_e\w(\abs{V}) = p_c \w(\abs{V})=
0.5 / \sqrt{\abs{V}-1}$. This is equivalent to the vertex degree scaling with
  $\sqrt{\abs{V}-1}$. Importantly, in the asymptotic limit $\abs{V} \to \infty$, it is
  $p_e\w(\abs{V}) > \ln\w(\abs{V}) / \abs{V}$, which ensures that the graph is connected
  almost surely (cf., e.g., \cite{spencer_ten_lectures_on_probabilistic_method}, more
  specifically see \footnote{ Let \en be the number of vertices, i.e., $n = \abs{V}$. In
    \cite{spencer_ten_lectures_on_probabilistic_method}, it is proven that, if $p_e(n) =
    \w(c + \ln n) / n$, then the graph $G$ is connected with probability $\ee^{-\ee^{-c}}$
  in the asymptotic limit $n \to \infty$ (a result by Erdös and R\'enyi). }). Notably, we
  see that the space costs for the approximated space-optimal schedule are close to the
  exact space-optimal schedule, at least up to the number of vertices for which we
  performed the exact optimisation (with much faster run-times, cf. \cref{f"runtime}). For
  larger graphs, however, this deviation would of course increase. Furthermore, the space
  cost of the time-optimal schedule obtained from the exact optimisation is lower than the
  space cost of the trivial time-optimal schedule. This is because the trivial
  time-optimal schedule does not try to measure first vertices with less neighbours.

For additional information about the numerical data, e.g., which probabilistic acceptance
function is chosen or for recorded run-times, see \cref{s"more_on_the_numerical_results}.

In general, the results show that appropriate scheduling of qubit initialisation and
measurement can potentially greatly reduce the space and time cost, depending on the
system parameters (especially compared to just naively initialising the whole graph).

Finding optimal schedules, however, is a hard problem, and will require optimised
algorithms for larger scales. For general quantum circuits it can not be expected to
find the optimal scheduling efficiently. The trivial time-optimal schedule, however, can
be obtained in polynomial time from appropriately tracked Pauli frames, by greedily
measuring the qubits as allowed by the induced time order. In an application, one can
imagine starting with the trivial order and then trying to optimise it with respect to
space, until a certain space cost is reached or alternatively stop the optimisation after
a certain timeout.

\begin{figure*}[hbt!]
  \center
  \includegraphics{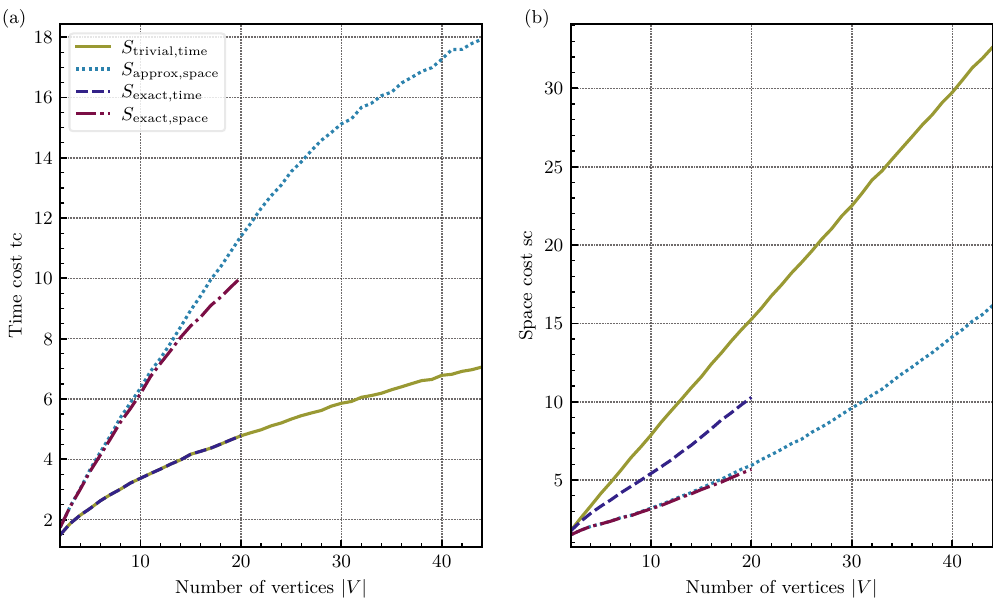}
  \caption[The time and space costs depending on the number of vertices.]{
    The time and space costs depending on the number of vertices $\abs{V}$. The edge and
    correction densities scale with $\abs{V}$: $p_e = p_c = 0.5 / \sqrt{\abs{V}-1}$. (a)
    shows the time cost $\mtc$ and (b) the space cost $\msc$ for different schedules.
    $S_{\text{trivial,time}}$ and $S_{\text{approx,space}}$ are the trivial time-optimal
    schedule and approximated space-optimal schedule as in \cref{f"density}.
    $S_{\text{exact,time}}$ and $S_{\text{exact,space}}$ are the according data of a full
    search for exact optimality, that is, $S_{\text{exact,time}}$ is a time-optimal
    schedule with the lowest possible space cost and $S_{\text{exact,space}}$ is a
    space-optimal schedule with the lowest possible time cost. These exact data points are
    only calculated for $\abs{V} \leq 20$ since the run-time becomes very long. It is
    $\mtc\w(S_{\text{exact,time}}) = \mtc\w(S_{\text{trivial,time}})$ and the costs for
    $S_{\text{approx,space}}$ are close to the costs of $S_{\text{exact,space}}$, however,
    $\msc\w(S_{\text{exact,time}})$ is fairly lower than
    $\msc\w(S_{\text{trivial,time}})$.
  }\label{f"nodes_main}
\end{figure*}

\section{Conclusion}

We presented a software library that provides the functionality to track Pauli gates
through a (Clifford) quantum circuit. The library is designed to be low level and generic,
allowing easy integration into other projects. The library is based on the isomorphism
between the Clifford group and the symplectic group. Tracking Pauli gates allows to reduce
the number of Pauli gates to the number of qubits or less. Furthermore, when tracking
Pauli corrections in \ac{mbqc}, the information can be used to calculate the strict
partial (time) order of the measurements. This allows us to perform scheduling
optimisations, which we tackled in the second part of the paper. We presented a framework
that covers the scheduling problem on an abstract level reduced to the underlying graph
and the time order. The numerical results we showed give indications on how much can be
gained with this optimisation technique.

\section{Software Availability}\label{s"software_availability}

The source code of the Pauli tracking library can be found in the
\hr{https://github.com/taeruh/pauli_tracker}[taeruh/pauli\_tracker] repository on GitHub
\cite{pauli_tracker_software}. The Rust library is published on
\hr{https://crates.io/crates/pauli_tracker}[crates.io] and the Python wrapper on
\hr{https://pypi.org/project/pauli-tracker/}[pypi.org] where you can also find links to
the documentation. For the source code of the MBQC scheduling project, see
\hr{https://github.com/taeruh/mbqc_scheduling}[taeruh/mbqc\_scheduling]
\cite{mbqc_scheduling_software}.

\section*{Acknowledgments}

We thank Samuel Elman, Thinh Le and Ryan Mann for helpful discussions. Jannis Ruh was
supported by the Sydney Quantum Academy, Sydney, NSW, Australia. The views, opinions,
and/or findings expressed are those of the author(s) and should not be interpreted as
representing the official views or policies of the Department of Defense or the U.S.
Government. This research was developed with funding from the Defense Advanced Research
Projects Agency [under the Quantum Benchmarking (QB) program under Awards No.
HR00112230007 and No. HR001121S0026 contracts].

\appendix

\section{More on the numerical results}\label{s"more_on_the_numerical_results}

\begin{figure*}[hbt!]
  \center
  \includegraphics{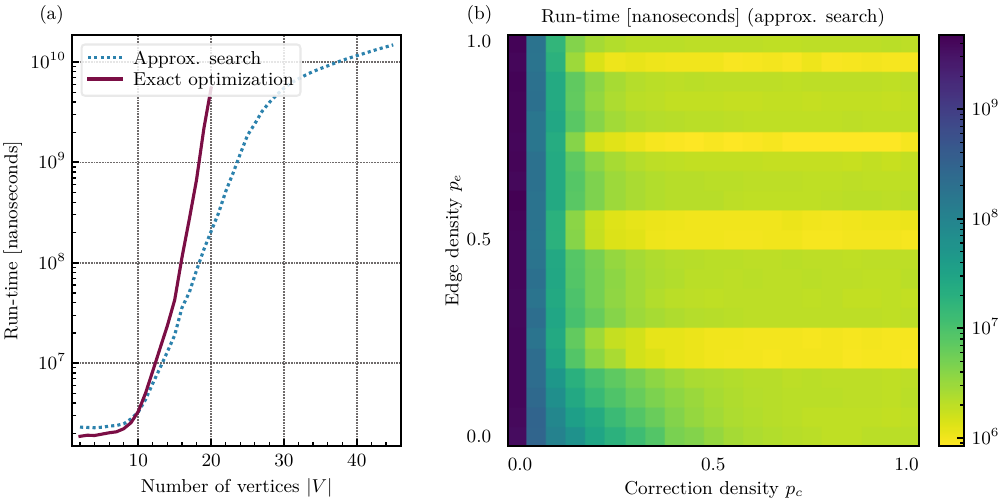}
  \caption[Recorded run-times of the calculations.]{
    Recorded run-times of the calculations. This is not a benchmark. In both plots the
    run-time is the average time it took to perform the optimisation, given a graph and a
    time order. (a) shows how the approximate algorithm performs compared to the exact
    algorithm, w.r.t. the number of vertices $\abs{V}$. Considering the logarithmic scale,
    it is magnitudes faster. Starting at approximately $30$ vertices, the timeout that
    scales quadratically with the number of vertices is reached, and the search returns
    early. (b) shows how the run-time depends on the edge density $p_e$ and the correction
    density $p_c$. It drastically increases for low correction densities. Vertically, for
    different edge densities, there are discontinuities. This is because for each selected
    edge density, the according data points, w.r.t. the correction density, were
    calculated on a different cluster node which were under different load (running the
    simulations on smaller grids locally on a single laptop does not show these
    discontinuities).
  }\label{f"runtime}
\end{figure*}

The code for the data generation and the numerical data can be found at the commits
\hr{https://github.com/taeruh/mbqc_scheduling/tree/6aaa4d8bf867075c3d90c91df87962e9e6c29377}[6aaa4d8]
and
\hr{https://github.com/taeruh/mbqc_scheduling/tree/040e28959cb87bba7b0ab64e97d2ad5b4fe1bebb}[040e289]
in
\hr{https://github.com/taeruh/mbqc_scheduling/tree/main/results}[taeruh/mbqc\_scheduling/results]
\cite{mbqc_scheduling_software}, respectively. The data capture all the information that
is needed to reproduce the data. However, an exact reproduction for the approximated
results is impossible. This is because the algorithm is multi-threaded and the tasks that
are sent to the threads communicate intermediate results between each other, which change
how the tasks continue their execution. While this is in theory deterministic, in practice
it depends on how the operating system schedules the threads and what the CPU is doing.

The probabilistic acceptance function that we used for the approximated space-optimal
schedules is given by
\begin{equation}
  p_{\mathrm{accept}}\w(\Delta, V, M) = \Theta\w(\abs{\Delta}) \abs{V}^2 \ee^{-
  \frac{\abs{V} \abs{V{\setminus}M}}{\abs{\Delta}^3 \w(\abs{M})}}
  \label{e"accept}
\end{equation}
where $V$ are the vertices in the graph, $M$ are the vertices that have been measured so far
(including the ones in the currently picked measurement set), $\Delta$ the difference
between the best memory and the required memory so far, and $\Theta$ is the Heaviside step
function (being $1$ for positive elements). The step function ensures that we only focus
on space-optimal schedules. For finding schedules somewhere between space optimal and
time optimal, i.e., a full space-time optimisation, one can choose a smoother function.
The exponential decay ensures that the run-time is sufficiently small for our experiments,
focusing on schedules with a rather large space improvement. In a real application, e.g.,
a real quantum compiler toolchain, one may choose a more optimised function, fitting to
the required needs, and optimise its parameters.

The probabilistic searches have a timeout which increases quadratically with the number of
qubits, i.e., we stop the search after $\order[\abs{V}^2]$ time and return the best
results that have been found so far. The effect of that can be seen in \cref{f"runtime}
(b), where the run-time drops down onto a quadratic curve for larger number of qubits,
however, note that the behaviour of the according curves in \cref{f"nodes_main} indicate
that the implied penalty on the costs is not too high.

The standard deviations, which we do not show for simplicity, on the plotted results are
relatively high. This is an artefact of combinatorics: For small number of vertices, the
standard deviation is just relatively high, and for larger number of vertices it becomes a
sampling issue. Because of that, we emphasise again that the results are rather
qualitative indications than quantitative benchmarks.

In \cref{f"runtime} we see the average run-time for the calculations. This is not a
benchmark, but only shows the recorded run-times for our calculations. The
calculations were performed on a cluster where the vertices were under different loads;
therefore, these plots are only qualitative indications. We see that the approximate
optimisation runs much faster than the exact search for larger graphs. In general, the
run-time heavily depends on the chosen acceptance function \cref{e"accept}. For small
correction densities, the run-time drastically increases, because the number of possible
measurement steps, per step, scales according to the ordered Bell numbers.

\bibliographystyle{quantum}

\bibliography{literature.bib}

\begin{acronym}[MBQC]
\acro{mbqc}[MBQC]{measurement-based quantum computation}
\acro{nisq}[NISQ]{noisy intermediate-scale quantum\acroextra{ (computing)}}
\acro{qec}[QEC]{quantum error correction}
\acro{simd}[SIMD]{single instruction, multiple data\acroextra{ (CPU instruction)}}
\end{acronym}

\end{document}